\documentclass{article}

\usepackage{microtype}
\usepackage{graphicx}
\usepackage{booktabs} %

\usepackage[pagebackref]{hyperref}

\usepackage[accepted]{icml2024}

\usepackage{amsmath}
\usepackage{amssymb}
\usepackage{mathtools}
\usepackage{amsthm}

\usepackage[textsize=tiny]{todonotes}

\usepackage{enumitem}

\usepackage{xspace}
\usepackage{bbm}
\usepackage{bm}
\input{mathlig}
\usepackage{mathtools}
\usepackage{relsize}
\usepackage{mathrsfs}
\usepackage{dsfont}
\DeclarePairedDelimiterX{\inp}[2]{\langle}{\rangle}{#1, #2}
\makeatletter
\newcommand*\bigcdot{\mathpalette\bigcdot@{.5}}
\newcommand*\bigcdot@[2]{\mathbin{\vcenter{\hbox{\scalebox{#2}{$\m@th#1\bullet$}}}}}
\makeatother

\newcommand{\muspace}{\mspace{1mu}}

\DeclareRobustCommand{\scond}{\mathchoice{\muspace\vert\muspace}{\vert}{\vert}{\vert}}
\mathlig{|}{\scond}

\DeclareRobustCommand{\discint}{\mathchoice{\mspace{-1.5mu}:\mspace{-1.5mu}}{\mspace{-1.5mu}:\mspace{-1.5mu}}{:}{:}}
\mathlig{::}{\discint}
\newcommand{\suchthat}{\mathchoice{\colon}{\colon}{:\mspace{1mu}}{:}}

\newcommand{\Cc}{\mathcal{C}}

\newcommand{\Fc}{\mathcal{F}}
\newcommand{\Gc}{\mathcal{G}}

\newcommand{\Mc}{\mathcal{M}}

\newcommand{\Vc}{\mathcal{V}}

\newcommand{\Xv}{{\bf X}}
\newcommand{\Yv}{{\bf Y}}
\newcommand{\Zv}{{\bf Z}}

\newcommand{\av}{{\bf a}}
\newcommand{\bv}{{\bf b}}

\newcommand{\mv}{{\bf m}}

\newcommand{\xv}{{\bf x}}
\newcommand{\yv}{{\bf y}}
\newcommand{\zv}{{\bf z}}

\newcommand{\sv}{{\bf s}}
\newcommand{\alphav}{{\boldsymbol{{\alpha}}}}

\newcommand{\muv}{{\bm{\mu}}}

\newcommand{\Yt}{{\tilde{Y}}}

\newcommand{\kt}{{\tilde{k}}}

\def\a{\alpha}

\def\d{\delta}

\DeclareMathOperator\E{\mathsf{E}}
\let\P\relax
\DeclareMathOperator\P{\mathsf{P}}

\newcommand{\Unif}{\mathrm{Unif}}
\newcommand{\Binom}{\mathrm{Binom}}
\newcommand\eg{e.g.,\xspace}
\newcommand\ie{i.e.,\xspace}
\def\textiid{i.i.d.\@\xspace}
\newcommand\iid{\ifmmode\text{ i.i.d. } \else \textiid \fi}

\newcommand{\Real}{\mathbb{R}}
\newcommand{\Natural}{\mathbb{N}}
\newcommand{\Integer}{\mathbb{Z}}

\newcommand{\ones}{\mathds{1}}

\def\mathllap{\mathpalette\mathllapinternal}
\def\mathllapinternal#1#2{%
  \llap{$\mathsurround=0pt#1{#2}$}}

\def\clap#1{\hbox to 0pt{\hss#1\hss}}
\def\mathclap{\mathpalette\mathclapinternal}
\def\mathclapinternal#1#2{%
  \clap{$\mathsurround=0pt#1{#2}$}}

\let\oldstackrel\stackrel
\renewcommand{\stackrel}[2]{\oldstackrel{\mathclap{#1}}{#2}}

\DeclarePairedDelimiterX{\infdivx}[2]{(}{)}{%
  #1\;\delimsize\|\;#2%
}

\renewcommand{\hbar}{h\mathllap{\overline{\vphantom{h}\hphantom{\rule{4.6pt}{0pt}}}\mspace{0.77mu}}}

\catcode`~=11 %
\newcommand{\urltilde}{\kern -.06em\lower -.06em\hbox{~}\kern .02em}
\catcode`~=13 %

\newcommand{\ind}[1]{\ones_{#1}}

\DeclarePairedDelimiterX{\norm}[1]{\lVert}{\rVert}{#1}
\DeclarePairedDelimiterX{\abs}[1]{\lvert}{\rvert}{#1}

\usepackage{xparse}

\newcommand*\diff{\mathop{}\!\mathrm{d}}

\let\oldpartial\partial
\renewcommand*{\partial}{\mathop{}\!\oldpartial}

\newcommand{\defeq}{\mathrel{\mathop{:}}=}

\newcommand{\wealth}{\mathsf{W}}

\newcommand{\hr}{\mathsf{HR}}
\newcommand{\chr}{\mathsf{CHR}}
\newcommand{\ob}{\mathbf{o}}
\newcommand{\nv}{\mathbf{n}}

\newcommand\numberthis{\addtocounter{equation}{1}\tag{\theequation}}

\usepackage{amsmath,amsthm,amssymb}
\newtheorem{theorem}{Theorem}

\newtheorem{lemma}[theorem]{Lemma}

\newtheorem{proposition}[theorem]{Proposition}
\theoremstyle{definition}

\usepackage{thmtools}
\declaretheoremstyle[
  headfont=\color{red}\normalfont\bfseries,
  bodyfont=\color{red}\normalfont\itshape,
]{colored}
\declaretheoremstyle[
  headfont=\color{blue}\normalfont\bfseries,
  bodyfont=\color{blue}\normalfont\itshape,
]{resolved}

\declaretheoremstyle[
  headfont=\color{blue}\normalfont\bfseries,
  bodyfont=\color{blue}\normalfont\itshape,
]{blue}

\declaretheoremstyle[
  headfont=\color{red}\normalfont\bfseries,
  bodyfont=\color{red}\normalfont\itshape,
]{red}

\newcommand{\kb}{\mathbf{k}}
\renewcommand{\E}{\mathop{\mathsf{E}}}

\newcommand{\ev}{\mathbf{e}}
\newcommand{\ov}{\mathbf{o}}
\newcommand{\pv}{\mathbf{p}}

\renewcommand{\emptyset}{\varnothing}
\renewcommand{\epsilon}{\varepsilon}
\renewcommand{\tilde}{\widetilde}
\renewcommand{\hat}{\widehat}

\newlength{\depthofsumsign}
\makeatletter
\newcommand*\dotp{\mathpalette\dotp@{.5}}
\newcommand*\dotp@[2]{\mathbin{\vcenter{\hbox{\scalebox{#2}{$\m@th#1\bullet$}}}}}
\makeatother

\newcommand{\bb}{\mathbf{b}}

\newcommand{\up}{\mathrm{up}}

\makeatletter
\newcommand\footnoteref[1]{\protected@xdef\@thefnmark{\ref{#1}}\@footnotemark}
\makeatother

\renewcommand{\ind}[1]{\mathbb{I}\{#1\}}

\usepackage[labelfont=it,font+=smaller]{caption}
\usepackage[labelfont=it]{subcaption}

\usepackage[scale=0.82]{FiraMono}

\icmltitlerunning{Gambling-Based Confidence Sequences for Bounded Random Vectors}

\begin{document}

\twocolumn[
\icmltitle{Gambling-Based Confidence Sequences for Bounded Random Vectors}

\icmlsetsymbol{equal}{*}

\begin{icmlauthorlist}
\icmlauthor{J. Jon Ryu}{mit}
\icmlauthor{Gregory W. Wornell}{mit}
\end{icmlauthorlist}

\icmlaffiliation{mit}{Department of EECS, MIT, Cambridge, Massachusetts, USA}

\icmlcorrespondingauthor{J. Jon Ryu}{\href{jongha@mit.edu}{jongha@mit.edu}}

\icmlkeywords{confidence sequence, time-uniform confidence sets, multinomial observation, count data analysis}

\vskip 0.3in
]

\printAffiliationsAndNotice{}  %

\begin{abstract}
A confidence sequence (CS) is a sequence of confidence sets that contains a target parameter of an underlying stochastic process at any time step with high probability. This paper proposes a new approach to constructing CSs for means of bounded multivariate stochastic processes using a general gambling framework, extending the recently established coin toss framework for bounded random processes. The proposed gambling framework provides a general recipe for constructing CSs for categorical and probability-vector-valued observations, as well as for general bounded multidimensional observations through a simple reduction. 
This paper specifically explores the use of the mixture portfolio, akin to Cover's universal portfolio, in the proposed framework and investigates the properties of the resulting CSs. Simulations demonstrate the tightness of these confidence sequences compared to existing methods. When applied to the sampling without-replacement setting for finite categorical data, it is shown that the resulting CS based on a universal gambling strategy is provably tighter than that of the posterior-prior ratio martingale proposed by Waudby-Smith and Ramdas.

\end{abstract}

\newcommand{\Simplex}{\Delta}
\newcommand{\Fv}{\mathbf{F}}
\newcommand{\Kv}{\mathbf{K}}
\newcommand{\kv}{\mathbf{k}}
\newcommand{\qv}{\mathbf{q}}
\newcommand{\Km}{\mathsf{K}}
\newcommand{\sprod}{{\textstyle\prod}}
\renewcommand{\kt}{\mathsf{KT}}
\renewcommand{\up}{\mathsf{UP}}
\renewcommand{\av}{\boldsymbol{\a}}
\newcommand{\Dir}{\mathrm{Dir}}
\allowdisplaybreaks

\section{Introduction}
Time-uniform (or anytime-valid) confidence sets, which is often called a \emph{confidence sequence} (CS), is a sequence of sets that contains a target parameter of the underlying stochastic process \emph{at any time} with high probability. 
The time-uniformity guarantee implies that a practitioner can make a decision by continuously acquiring new data in a sequential manner, without the concern of the so-called p-hacking problem with the classical (non-time-uniform) confidence sets~\citep{Ramdas--Ruf--Larsson--Koolen2020}.
The notion of time-uniform guarantee dates back to \citet{Hoeffding1963,Darling--Robbins1967,Lai1976}, but the idea has regained attention only recently with the new modern applications such as A/B testing~\citep{Evidently2023} and election audits~\citep{Waudby-Smith--Stark--Ramdas2021}.
As tight confidence sequences can directly lead to saving various resources in practical sequential decision making applications, there has been a surge of interest in the research community, \eg \citep{Waudby-Smith--Ramdas2020b,Howard--Ramdas--McAuliffe--Sekhon2021,Orabona--Jun2021}, to name a few. 

While the field is rapidly evolving with new theories and constructions for CS~\citep{Waudby-Smith--Ramdas2020b,Howard--Ramdas--McAuliffe--Sekhon2021,Jun--Orabona2019,Orabona--Jun2021,Manole--Ramdas2023,Shekhar--Ramdas2023}, we remark that most of the development has been especially for scalar-valued processes until recently; a few exceptions include \citep{Waudby-Smith--Ramdas2020a} for categorical observation in without replacement sampling, and \citep{Whitehouse--Wu--Ramdas2023,Chugg--Wang--Ramdas2023} for unbounded vectors.

The goal of this paper is to develop a set of new methods for constructing CSs for \emph{bounded, vector-valued} stochastic processes. 
Note that there exist natural bounded vector-valued data in practice, a canonical example of which is a categorical data.
For bounded scalar-valued processes, \citet{Waudby-Smith--Ramdas2020b} proposed a coin-toss-based framework, studying several gambling strategies and properties of the resulting CSs. 
In the same framework, \citet{Orabona--Jun2021} considered applying the celebrated algorithm of Cover's universal portfolio~\citep{Cover1991} in information theory, proposed its tight outer approximations, and established analytical properties of their algorithms.
A subsequent work of \citet{Ryu--Bhatt2022} further explored the idea of universal gambling, arguing that the coin-toss framework can be understood as the continuous two-horse race, which admits more natural interpretation and closes the conceptual gap between CS and universal gambling.
\citet{Shekhar--Ramdas2023} recently established near-optimality of gambling-based CSs.

Given the elegance of the coin-toss framework for bounded, scalar-valued processes and its empirical tightness, especially in the small-sample regime~\citep{Orabona--Jun2021}, it is tempting to ask whether an analogous gambling framework exists for bounded, vector-valued processes. In this paper, we provide an affirmative answer: starting from the two-horse race formulation of \citet{Ryu--Bhatt2022}, we extend the existing coin-toss framework~\citep{Waudby-Smith--Ramdas2020b} to a general gambling framework for bounded, vector-valued processes. The key observation that leads to the multivariate extension is that the coin-toss framework essentially constructs confidence sequences for binary probability vectors, each of which happens to be characterized by a single number.
Going forward, we explore the resulting CSs from this framework when we use a universal gambling strategy such as \citet{Cover1991}'s universal portfolio as in \citep{Orabona--Jun2021,Ryu--Bhatt2022} for scalar-valued processes, and study their theoretical properties and empirical performance. 

The rest of the paper is organized as follows.
In Section~\ref{sec:framework}, we gradually build up the idea for the general gambling framework, by starting from the standard definitions, basic tools, and the coin-toss framework for scalar-valued processes.
In the next two sections, we study the resulting CSs for the categorical observation (Section~\ref{sec:categorical}) and probability-valued observation (Section~\ref{sec:probability}), showing the nice properties of CSs resulting from universal gambling strategies.
In Section~\ref{sec:general}, we briefly describe how the general, bounded vector processes can be handled as probability-vector valued observations.
We conclude with remarks in Section~\ref{sec:conclusion}.

\section{General Gambling Framework}
\label{sec:framework}
In this section, we illustrate the key idea of the paper on how to construct a confidence sequence for vector-valued observations via gambling.

\subsection{Preliminaries}
We start with a few definitions and notation.
We use a boldface variable, \eg $\xv$, to denote a vector-valued observation.
We use a capital letter $X_t$ (or $\Xv_t$) to denote a random variable (or a random vector) and a small letter $x_t$ (or $\xv_t$) for its realization.
We use the shorthand notation $\xv_i^j\defeq(x_i,\ldots,x_j)$ and $\xv^t\defeq \xv_1^t$.
For an underlying stochastic process $\Yv^t$,
let $\Fc_t\defeq\sigma(\Yv^t)$ be the sigma-field generated by $\Yv^t$, where $\Fc_0$ is the trivial sigma-field.
We call the sequence $(\Fc_t)_{t=0}^\infty$ the \emph{canonical filtration}.

We assume that there exists an underlying, bounded vector-valued stochastic process $(\Yv_t)_{t=1}^\infty$ such that $\E[\Yv_t|\Fc_{t-1}]=\muv$ for any $t\ge1$ for some $\muv\in\Real^K$.
Our goal is to construct \emph{time-uniform} confidence sets, so that the target mean parameter $\muv$ is to be contained in the sequence at any time, with high probability.
Formally, we wish to find a sequence of set-valued functions, each of which maps the observations $\yv^{t}$ until then to a \emph{convex} subset $\Cc(\yv^{t};\delta)\subset \Real^{K}$, such that
\[
\P\{
\muv\in \Cc(\Yv^{t};\delta), \forall t\ge 1
\}\ge 1- \d,
\]
for any given desired parameter $\d\in(0,1)$.

Our development follows the standard technique based on the following inequality by \citet{Ville1939} for supermartingales,  which is a cornerstone of the most known confidence sequences. 
Different methods plug in different {supermartingales} given their own problem settings and considerations.
A stochastic process $(W_t)_{t=0}^\infty$ is said to be \emph{supermartingale} if $(W_t)_{t=0}^\infty$ is adapted to the canonical filtration and satisfy
\[
\E[W_t|\Fc_{t-1}]\le W_{t-1}
\]
for every $t\ge 1$. It is said to be \emph{martingale}, if the inequality holds with equality for any $t\ge 1$.

\begin{theorem}[Ville's inequality]
\label{thm:ville}
For a nonnegative supermartingale sequence $(W_t)_{t=0}^{\infty}$ with $W_0>0$, for any $\d>0$, %
\[
\P\Bigl\{\sup_{t\ge1} \frac{W_t}{W_0} \ge \frac{1}{\d}\Bigr\}\le \d.
\]
\end{theorem}
A general recipe to derive confidence sequences based on Ville's inequality can be briefly summarized as follows.
First, construct a sequence $(W_t(\yv^t;\mv))_{t=1}^\infty$ as a function of $\yv^t$ and a candidate parameter $\mv$ at each time $t\ge 1$.
Define
\begin{align}
\Cc(\yv^t;\d)\defeq \Bigl\{\mv\suchthat \frac{W_t(\yv^t;\mv)}{W_0(\emptyset;\mv)}<\frac{1}{\d} \Bigr\}.
\end{align}
If the induced stochastic process $(W_t(\Yv^t;\mv))_{t=1}^\infty$ whose randomness comes from $\Yv^t$ is (super)martingale for $\mv=\muv$, then it readily follows from Ville's inequality that $(\Cc(\Yv^t;\d))_{t=1}^\infty$ is a confidence sequence for $\muv$ of coverage $1-\d$.\footnote{We note in passing that the running intersection $(\cap_{i=1}^t \Cc_i(\yv^i;\d))_{t=1}^\infty$ is always tighter or equal to the original CS and also a valid CS, and thus this operation is always implicitly assumed.}
Since we expect the confidence sequence to be shrinking to the singleton that contains $\muv$ as fast as possible while maintaining the desired coverage, the sequence $(W_t(\yv^t;\mv))_{t=1}^\infty$ needs to be carefully constructed.
Ideally, the stochastic process $(W_t(\yv^t;\mv))_{t=1}^\infty$ should diverge quickly for all $\mv\neq \muv$, so that any $\mv\neq\muv$ can be excluded from $(\Cc(\Yv^t;\d))_{t=1}^\infty$, eventually letting it converge to the singleton $\{\muv\}$.
In summary, the (super)martingale property of the process for $\mv=\muv$ guarantees the constructed sequence to be a valid confidence sequence, and
the speed of the divergence of the stochastic process for $\mv\neq \muv$ will govern the speed of concentration to the target parameter.

\subsection{Coin-Toss Framework for $[0,1]$-Valued Processes}
One natural and elegant way of constructing a martingale sequence is to consider the wealth process of \emph{gambling}.
The high-level idea is that if a gambling is \emph{subfair} (\ie statistically not favorable to gamblers), then a gambler's wealth as a stochastic process must be supermartingale, that is, the gambler cannot earn money out of the gambling.
The idea of using gambling has been extensively studied by a line of recent works for bounded, scalar-valued stochastic processes~\citep{Waudby-Smith--Ramdas2020b,Orabona--Jun2021,Ryu--Bhatt2022,Shekhar--Ramdas2023}.
Below, we briefly overview the existing formulation and motivate our idea for the extension to multivariate processes.

We start from the standard betting framework introduced in \citep{Waudby-Smith--Ramdas2020b} for a conditionally mean-$\mu$ stochastic process $(Y_t)_{t=1}^\infty$ for some $\mu\in(0,1)$. 
\citet{Waudby-Smith--Ramdas2020b} introduced the definition of a \emph{capital process} $\Km_t(m)$ for any $m\in(0,1)$ as
\begin{align}
\label{eq:standard_convention}
\Km_t(y^t;m)\defeq \prod_{i=1}^t (1+ \lambda_i(m) \cdot (y_i-m)),
\end{align}
where $(\lambda_t(m))_{t=1}^\infty$ is a $[-1/(1-m),1/m]$-valued predictable sequence, \ie $\lambda_t(m)$ only depends on $Y^{t-1}$, but not on the future $Y_t^\infty$.
It is clear that the capital process is martingale for $m=\mu$, since
$\E[\Km_t(Y^t;\mu)|\Fc_{t-1}] = \Km_t(Y^{t-1};\mu)\cdot (1+\lambda_i(m) \cdot (\E[Y_t|\Fc_{t-1}]-\mu))= \Km_t(Y^{t-1};\mu)$.
Here, $(Y_t)_{t=1}^\infty$ can be viewed as the stochastic outcomes of coin betting, and the predictable sequence $(\lambda_t(m))_{t=1}^\infty$ can be understood as a causal betting strategy. \citet{Waudby-Smith--Ramdas2020b} studied various different betting strategies and their properties in terms of the resulting confidence sequences.
We remark, however, that it is nontrivial to extend this betting framework to vector-valued processes, as also pointed out in the discussion of \citet{Li--Li--Dai2023}.
One critical complication in this convention is that the range of the bet $\lambda_t(m)$ depends on the candidate parameter $m$.

More recently, \citet{Ryu--Bhatt2022} proposed an equivalent, yet alternative formulation of the betting framework, connecting it to the standard language of gambling in information theory~\citep[Ch.~6, 16]{Cover--Thomas2006}. 
The alternative formulation is based on the (continuous) two-horse race interpretation of the capital process~\eqref{eq:standard_convention}. 
The rule of the standard horse race is as follows.
There are two horses in the race, say horse 1 and horse 2, associated with \emph{odds} $o_1>0$ and $o_2>0$, respectively. These numbers indicate that if the horse $i$ wins, a \$1 bet on the horse results in a payoff of \$$o_i$. 
At each round $t\ge 1$, a gambler bets over the (yet unseen) outcome of the horse race with two horses. The bet is characterized by a single number $b_t\in[0,1]$, which is to distribute the current wealth $\wealth_t$ of the gambler to the two horses as \$$b_t\wealth_t$ and \$$(1-b_t)\wealth_t$.
After the race, the per-round multiplicative gain of the gambler's wealth can be written as 
\begin{align}
\label{eq:standard_two_horse_race}
(o_1b_t)^{Y_t}(o_2(1-b_t))^{1-Y_t},    
\end{align}
where $Y_t\defeq \ind{Z_t=1} \in \{0,1\}$ and $Z_t\in\{1,2\}$ denotes the index of the winning horse.
The continuous variation of the horse race allows the outcome $Y_t$ to be $[0,1]$-valued (hence continuous) and defines the per-round gain as
\begin{align}
\label{eq:continuous_two_horse_race}
o_1b_tY_t + o_2(1-b_t)(1-Y_t).
\end{align}
Note that \eqref{eq:continuous_two_horse_race} is equivalent to \eqref{eq:standard_two_horse_race} for discrete outcomes $Y_t\in\{0,1\}$.
\citet{Ryu--Bhatt2022} noted that the capital process~\eqref{eq:standard_convention} can be viewed as the wealth process from the continuous two-horse race setting with odds $o_1=\frac{1}{m}$, $o_2=\frac{1}{1-m}$ and a strategy $(b_t)_{t=1}^{\infty}$, as the per-round gain can be written as
\[
\frac{1}{m} b_t y_t + \frac{1}{1-m} (1-b_t)(1-y_t)
=1+\lambda_t(m) (y_t-m),
\]
if we define the {\emph{scaled bet}}
\[
\lambda_t(m) \defeq \frac{b_t}{m(1-m)}-\frac{1}{1-m}\in \Bigl[-\frac{1}{1-m}, \frac{1}{m}\Bigr].
\numberthis
\label{eq:signed_betting}
\]
Compared to \eqref{eq:standard_convention}, this reformulation separates out the amount of bet $b_t$ from $m$ and thus the range of the bet is simply $[0,1]$, being $m$-independent.

\subsection{Extension for Multivariate Processes: Overview}
We are now ready to introduce an extension of the betting framework for vector-valued processes following the convention of \citet{Ryu--Bhatt2022}. 
To proceed, we first introduce a general abstract formulation of the gambling as a multiplicative game.
Without loss of generality, we assume a unit dollar for the initial wealth of a gambler, \ie $\wealth_0=\$1$.
We assume that $\Mc\subseteq\Real_+^K$ is given as a set of \emph{odds vectors}.
For each round $t\ge 1$, a gambler chooses a bet $\bb_t=\bb_t(\xv^{t-1})\in\Delta^{K-1}$ as a function of the previous odds vectors $\xv^{t-1}\in\Mc^{t-1}$, where
we use $\Delta^{K-1}\defeq \{\pv\in\Real_{\ge0}^K \suchthat p_1+\ldots+p_K=1\}$ to denote the $(K-1)$-dimensional probability simplex in $\Real^K$.
After each round, the odds vector $\xv_t\in\Mc$ is revealed, and the gambler's wealth is multiplied by $\bb_t^\intercal\xv_t$,
and thus the multiplicative gain after round $t$ is
\[
\wealth_t(\xv^t) =  \prod_{i=1}^t \bb_i^\intercal\xv_i.
\]
As alluded to earlier, it is easy to show that the wealth process $(\wealth_t)_{t=1}^\infty$ attained by any causal strategy is a supermartingale, if the sequence of odds $(\Xv_t)_{t=1}^{\infty}$ \emph{subfair}, \ie $\E[\Xv_t|\Fc_{t-1}]\le\ones_K$ for any $t\ge 1$, where the inequality holds coordinatewise for $\ones_K\defeq [1,\ldots,1]\in\Real^K$.

\newcommand{\hrtwo}{$\hr_2$}
\newcommand{\chrtwo}{$\chr_2$}
\newcommand{\hrk}{$\hr_K$}
\newcommand{\chrk}{$\chr_K$}
We remark that this abstract setting subsumes the two-horse race \hrtwo{} and continuous two-horse race \chrtwo{} as special cases, which correspond to certain restricted sets of odds vectors $\Mc=\{o_1\ev_1,o_2\ev_2\}$ and $\Mc=\{o_1y\ev_1 + o_2(1-y)\ev_2\suchthat y\in [0,1]\}$, respectively.
For $K\ge2$ horses, we can consider the $K$-horse race \hrk{} and continuous $K$-horse race \chrk{}, where $\Mc=\{o_1\ev_1,\ldots,o_K\ev_K\}$ and $\Mc=\{o_1y_1\ev_1 +\ldots + o_Ky_K\ev_K\suchthat \yv\in \Simplex_K\}$ are assumed, respectively.
The most general scenario is when the odds vector $\xv_t$ can take any value from the positive orthant, \ie $\Mc=\Real_{>0}^K$. In this case, the gambling becomes equivalent to the \emph{stock market investment}~\citep{Cover1991}.

The key idea of this paper is to observe that one can consider \hrk{} and \chrk{} to derive CSs for $K$-categorical data and $\Delta^{K-1}$-valued stochastic processes.
Recall that \hrtwo{} and \chrtwo{} can be used to derive CSs of the mean parameter for $\{0,1\}$-valued and $[0,1]$-valued stochastic processes, respectively.
The extension is natural, once a single number $Y_t$ is mapped into the form of a binary probability vector $[Y_t,1-Y_t]$ lying in the probability simplex $\Delta^1$.
In general, for a $\Delta^{K-1}$-valued process $(\Yv_t)_{t=1}^\infty$ such that $\E[\Yv_t|\Fc_{t-1}]=\muv$, the gambling with odds vector sequence $(\Xv_t)_{t=1}^\infty$, which is defined as
\[
\Xv_t \defeq \frac{\Yv}{\muv}\defeq \Bigl(\frac{Y_{t1}}{\mu_1},\ldots, \frac{Y_{tK}}{\mu_K}\Bigr),
\numberthis\label{eq:odds_chr}
\]
is fair,
since $\E[\Xv_t|\Fc_{t-1}]=\ones$ by construction.
Note that the odds vector~\eqref{eq:odds_chr} corresponds to \hrk{} if $\Yv_t\in\{\ev_1,\ldots,\ev_K\}$ (\ie multinomial, or categorical observations) and \chrk{} if $\Yv_t\in\Simplex_K$ in general.
The embedding of the mean vector parameter in \eqref{eq:odds_chr} is the essence of the multivariate extension.

In the rest of this paper, we build up this idea more concretely.
In Section~\ref{sec:categorical}, we start by categorical observations, where our key ideas and the application of universal gambling strategies are illustrated with the simple setting of \hrk{}. In this special case, we also show that when reduced to the sampling without replacement setting, the derived CS yields a provably tighter CS than that of \citep{Waudby-Smith--Ramdas2020a}.
We then present its natural continuous extension for probability-vector-valued observations in Section~\ref{sec:probability}, where the underlying gambling becomes \chrk{}, a special case of the stock market investment. 
We also present how a general $[0,1]^{K-1}$-valued observations can be handled by a simple reduction to \chrk{}.
All proofs are deferred to Appendix~\ref{app:proofs}.

For a concise treatment of the confidence sequences to be developed, we define a few additional notations.
Let $\Vc_K\defeq \{\ev_1,\ldots,\ev_K\}$ denote the set of vertices of the probability simplex $\Delta^{K-1}$, \ie the set of one-hot vectors.
Let $\Gc_{K,t}\defeq\{\kv\in\Integer_{\ge 0}^{K}\suchthat k_1+\ldots+k_K=t\}$ for $t\ge 0$ denote the set of all nonegative-integer-valued vectors each of whose sum of the coordinates is equal to $t$.
For example, $\Gc_{K,0}=\{\boldsymbol{0}_K\}$ and $\Gc_{K,1}=\Vc_K$.
Since these objects will be frequently referred, we summarize them in Table~\ref{tab:notation}.

\begin{table}[tb]
    \centering
    \caption{Summary of notations.}
    \begin{tabular}{cl}
    \toprule
    Notation & Definition\\
    \midrule
        $\Vc_K$ & $\{\ev_1,\ldots,\ev_K\}$ \\
        $\Gc_{K,t}$ & $\{\kv\in\Integer_{\ge 0}^{K}\suchthat k_1+\ldots+k_K=t\}$ \\
    \bottomrule
    \end{tabular}
    \label{tab:notation}
\end{table}

\section{Categorical Observations}
\label{sec:categorical}
We first consider obtaining a confidence sequence for multinomial (or discrete categorical) observations by considering the standard $K$-horse race \hrk{}.
Hereafter, we refer to the horse race with odds $\ov\defeq(o_1,\ldots,o_K)$ by \hrk{}$(\ov)$.

For a $K$-categorical stochastic process $(Z_t)_{t=1}^{\infty}$ such that $\E[\ind{Z_t=j}|\Fc_{t-1}]=\mu_j$ for each $j\in[K]$ for some $\muv\in\Delta^{K-1}$, 
we can define its vector equivalent $\Zv_t\in\Vc_K\subset\Simplex_K$, where $Z_{tj}\gets \ind{Z_t=j}$ for each $j\in [K]$.
As hinted in the previous section, we consider the gambling \hrk{}$(\mv^{-1})$, where the odds vector corresponds to
\[
\Xv_t
\defeq {\mv^{-1}}\odot {\Zv_t}
=\Bigl(\frac{\ind{Z_{t1}=1}}{m_1}, 
\ldots,
\frac{\ind{Z_{tK}=1}}{m_K}\Bigr).
\]
Let $\wealth(\Zv^t;\ov)$ be the cumulative wealth from the horse race \hrk{}$(\ov)$. Since $\E[\Zv_t|\Fc_{t-1}]=\muv$, $\E[\Xv_t|\Fc_{t-1}]=\ones$ for $\mv=\muv$, and thus the resulting wealth process is martingale.
For $0<\d<1$, define
\[
\Cc(\Zv^t;\d)\defeq \Bigl\{
\mv\in\Simplex_K
\suchthat {\wealth(\Zv^t;\mv^{-1})} <\frac{1}{\d}
\Bigr\}.
\numberthis\label{eq:kt_cs}
\]
The reasoning in the previous section readily implies:
\begin{theorem}
\label{thm:hr}
For any causal gambling strategy, $(\Cc(\Zv^t;\d))_{t=1}^{\infty}$ is a CS with level $1-\d$.
\end{theorem}

\subsection{Krichevsky--Trofimov CS}
Since any causal gambling strategy results in a valid confidence sequence, a natural question is then which gambling strategy should be used to derive tight confidence sequences.
In this section, we consider applying a celebrated yet simple \emph{mixture} betting strategy from universal gambling.
The resulting CS is easy to compute, naturally convex, and always nonvacuous even after observing one sample. 
As a shorthand notation, for two vectors $\mathbf{a}$ and $\bv$ of same dimension $K$, we denote $\mathbf{a}^{\bv}\defeq \prod_{j=1}^K a_j^{b_j}$.

The celebrated idea of universal gambling from information theory is essentially to utilize a continuous mixture of the wealth processes of constant bettors.
Consider the cumulative wealth attained by a constant bettor $\bv$ from \hrk{}$(\ov)$:
\begin{align*}
{\wealth^{\bv}(\zv^t;\ov)}
&\defeq \prod_{i=1}^t \prod_{j=1}^K (o_jb_{j})^{\ind{z_{ti}=1}}
=(\ov\odot\bv)^{\kv(\zv^t)}.
\end{align*}
Here, we define the cumulative count vector as
\[
\kv(\zv^t)\defeq 
\Bigl(\sum_{i=1}^t \ind{z_{ij}=1}\Bigr)_{j\in[K]}
=\sum_{i=1}^t \zv_i.
\]
While any mixture of this wealth process induces a valid wealth process,
we specifically consider the so-called Krichevsky--Trofimov (KT) strategy~\citep{Krichevsky--Trofimov1981},
which takes a continuous mixture of the wealth of constant bettors with respect to the the Dirichlet distribution $\Dir(\bv;\av)\defeq \frac{1}{B(\av)}\prod_{i=1}^K b_i^{\a_i-1}$. Here, $B(\av)\defeq (\prod_{i=1}^K\Gamma(\a_i))/\Gamma(\sum_{i=1}^K\a_i)$ denotes the multivariate beta function and $\Gamma(\a)$ denotes the gamma function. 
That is, 
\begin{align*}
{\wealth^{\kt}(\zv^t;\mv^{-1})}
&\defeq 
\int_{\Delta^{K-1}} {\wealth^{\bv}(\zv^t;\mv^{-1})} \Dir(\bv;\av)\diff \bv\\
&=
\mv^{-\kv(\zv^t)}
\frac{B(\kv(\zv^t)+\av)}{B(\av)}.
\numberthis
\label{eq:kt_wealth}
\end{align*}
The KT mixture specifically sets $\av=\frac{\ones_K}{2}$.
This specific choice is the Jeffreys prior for multinomial distributions~\citep{Yang--Berger1996}. 
\citet{Krichevsky--Trofimov1981} proved, in the equivalent language of universal coding, that the KT mixture wealth can achieve the optimal wealth achieved by the best constant bettor in hindsight with a minimax optimal rate in the regret in log wealth ratios under a stochastic market assumption.
It is also later shown that the optimality even holds for individual sequences~\citep{Xie--Barron2000}.

This wealth process is easy to compute for any given $\mv$, since it is only a function of the count vector $\kv(\zv^t)$, which can be viewed as the sufficient statistics.
We can further manipulate this expression to better understand the behavior of this wealth process.
First, we define the empirical mean vector $\hat{\muv}_t\defeq \frac{1}{t}\sum_{i=1}^t \zv_i = \frac{\kv(\zv^t)}{t}$. 
Second, we define $q^{\kt}(\xv)\defeq \frac{B(\xv+\av)}{B(\av)}$ for $\xv\in\Real_{\ge 0}^K$. 
Then, the corresponding confidence set $\Cc^{\kt}(\zv^t;\d)$ can be written as
\begin{align*}
&\Cc^{\kt}(\zv^t;\d)
\\&
=
\Bigl\{\mv\suchthat%
\hat{\muv}_t^\intercal\log\frac{1}{\mv} 
<
\frac{1}{t}\log\frac{1}{\d} + \frac{1}{t}\log\frac{1}{q^{\kt}(t\hat{\muv}_t)}\Bigr\}\\
&=
\Bigl\{
\mv\suchthat%
D(\hat{\muv}_t~\|~\mv) < \frac{1}{t}\log\frac{1}{\d} + \frac{1}{t}\log\frac{e^{-t H(\hat{\muv}_t})}{q^{\kt}(t\hat{\muv}_t)}\Bigr\},
\end{align*}
where $D(\pv~\|~\qv)\defeq \pv^\intercal\log\frac{\pv}{\qv}$ denotes the Kullback--Leibler divergence between $\pv,\qv\in\Delta^{K-1}$ and $H(\pv)\defeq -\pv^\intercal\log \pv$ denotes the Shannon entropy for a probability vector $\pv\in\Delta^{K-1}$.
The following properties are now immediate.
\begin{theorem}[KT CS]
\label{thm:kt}
Pick any $\d\in(0,1)$.
\begin{enumerate}[label={(\alph*)}]
\item (Convexity) $\Cc^{\kt}(\zv^t;\d)\subset\Delta^{K-1}$ is always convex.
\item (Always-non-vacuous guarantee) If $\mv\in\Cc^{\kt}(\zv^t;\d)$, 
\[
m_j > \frac{k_j(\zv^t)}{t}(\d q^{\kt}(\kv(\zv^t)))^{\frac{1}{t}}
\]
for any $j\ge 1$. 
In particular, for any $j\in[K]$, $\Cc^{\kt}(\zv^t;\d)$ is strictly bounded away from the hyperplane $\{\mv\suchthat m_j=0\}$ as soon as $\min_{i\in[t]} z_{ij}=1$.
\item (Large-sample behavior) 
For $t$ sufficiently large,
$\Cc^{\kt}(\Zv^t;\d)$ asymptotically behaves as
\begin{align*}
\Bigl\{\mv\suchthat D(\hat{\muv}_t~\|~\mv) < \frac{1}{t}\log\frac{1}{\d} + \frac{K-1}{2t} \log t+o(1) \Bigr\},
\end{align*}
where $o(1)$ is a vanishing term as $t\to\infty$.
\end{enumerate}
\end{theorem}

\begin{figure*}
\centering
\includegraphics[width=\textwidth]{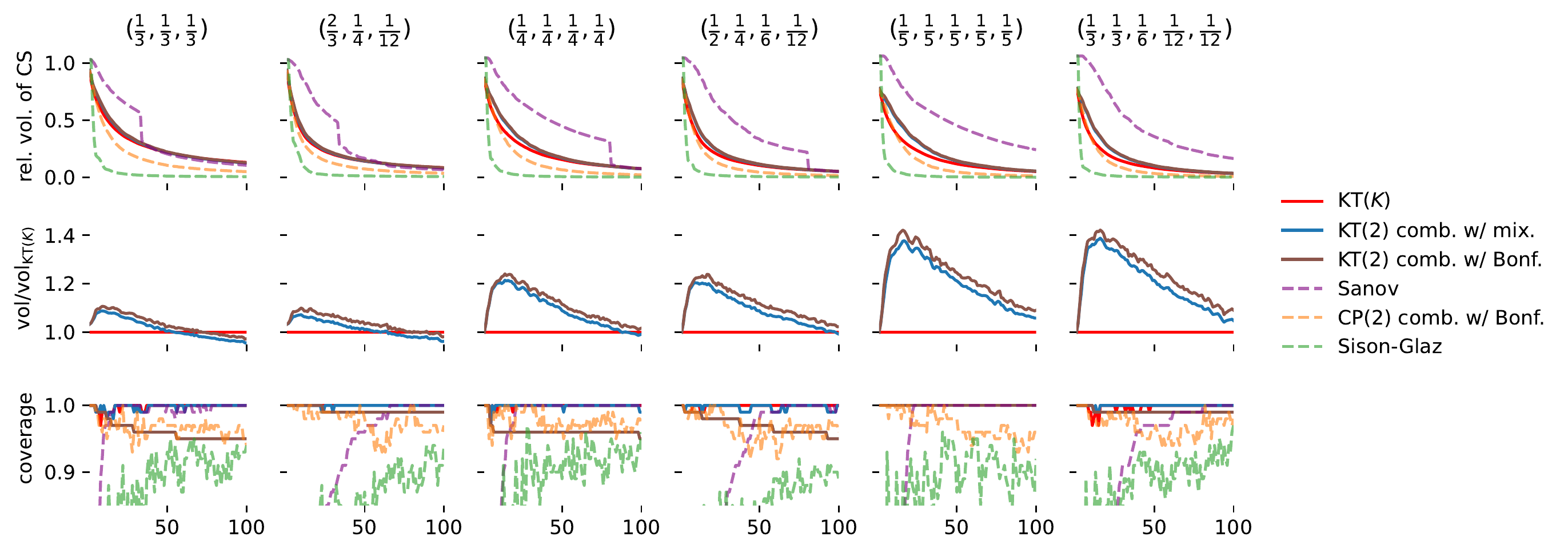}
\caption{Performance of time-uniform confidence sets (solid lines) and non-time-uniform confidence sets (dashed lines) with respect to \iid categorical data with mean vector $\muv\in\Simplex^{K-1}$ (indicated as the titles), in terms of the relative volume of the confidence sets (first row) and the per-time-step coverage (third row). The second row visualizes how the baselines $\mathsf{KT(2)~w/~mix.}$ and $\mathsf{KT(2)~w/~Bonf.}$ are compared to $\mathsf{KT}(K)$.}
\label{fig:ktcs}
\end{figure*}

\subsubsection{Baselines}
Instead of the multi-horse-race approach,
one can construct a confidence set for categorical data by properly \emph{aggregating} confidence intervals constructed for each coordinate, \eg using the existing gambling techniques for scalar-valued processes.
Below, we empirically demonstrate the benefit of the proposed approach compared to such baselines.

More concretely, for each coordinate $j\in[K]$, we can construct a confidence set with level $1-\d'$ for $\mu_j$, say $\Cc_{\delta'}^{(j)}$, based on a binary sequence $(\ind{z_{ij}=1})_{i=1}^t$.
To construct a single confidence set for $\muv$ from such $K$ confidence sets,
one straightforward choice is the well-known \emph{Bonferroni correction}, which takes the intersection of $K$ CSs constructed with level $1-\d'$ with $\d'=\frac{\d}{K}$ for each coordinate: 
\[
\Cc_\delta^{\sf{bonf}}\defeq \bigcap_{j=1}^K \Cc_{\frac{\delta}{K}}^{(j)}.
\]
By a union bound, it readily follows that the resulting CS is a CS with level $1-\d$.

For gambling-based (or martingale-based) CSs, there exists yet another natural way to aggregate CSs from different wealth sequences. Suppose that we are given $K$ wealth sequences $(\wealth_t^{(j)})_{t=0}^\infty$ for $j\in[K]$, each of which induces a level-$(1-\d')$ CS as
\[
\Cc_{\delta'}^{(j)}\defeq \Bigl\{\mv\suchthat \wealth_t^{(j)}(\mv)\le \frac{1}{\delta'}\Bigr\}.
\]
We can then aggregate the $K$ sets by a \emph{mixture wealth}, \ie
\[
\Cc_\delta^{\sf{mix}}\defeq \Bigl\{\mv\suchthat \frac{1}{K}\sum_{j=1}^K \wealth_t^{(j)}(\mv)\le \frac{1}{\delta}\Bigr\}.
\]
Then, it is straightforward that the Bonferroni correction is provably not smaller than the mixture-based aggregation:
\begin{proposition}
$\Cc_\delta^{\sf{mix}}\subseteq \Cc_\delta^{\sf{bonf}}$.
\end{proposition}
\begin{proof}
If $\mv\in \Cc_\delta^{\sf{mix}}$, then $\wealth_t^{(j)}(\mv)\le \sum_{j'=1}^K \wealth_t^{(j')}(\mv) \le \frac{K}{\delta}$ for every $j\in[K]$, which implies $\mv\in \Cc_\delta^{\sf{bonf}}$.
\end{proof}
We note that a concurrent work by \citet[Lemma~4]{Cho--Gan--Kallus2024} essentially discusses these two baselines and argues the tightness of the mixture-based aggregation, while we independently rediscovered this during the revision of our manuscript.

\subsubsection{Simulation\footnote{We have open-sourced the implementation of the proposed confidence sequences and provided the codes to reproduce the simulation results online: \url{https://github.com/jongharyu/confidence-sequence-via-gambling}.}}
To demonstrate the performance of the KT CS (denoted as $\mathsf{KT}(K)$, we run simulations with \iid categorical observations of length $T=100$, with different mean vectors of dimension $K\in\{3,4,5\}$; see Fig.~\ref{fig:ktcs}, where the title of each column indicates the underlying mean. Each experiment was run with 100 random realizations and $\d=0.05$ was used. In the first row, we plot the relative volume of the CS over $\Delta^{K-1}$. 
As another CS baseline, we consider applying the mixture aggregation and Bonferroni correction to combine $\mathsf{KT}(2)$'s induced by each coordinate $j\in[K]$ with $\d'=\d/K$; check $\mathsf{KT(2)~w/~mix.}$ and $\mathsf{KT(2)~w/~Bonf.}$ in the figure.
Note that, as argued above, mixture-based aggregations are tighter than Bonferroni, though marginally.
While the difference is not significant in $K=3$, we observe that the multivariate version $\mathsf{KT}(K)$ exhibits much tighter behavior as $K$ increases, as visualized in the second row. Interestingly, however, the baselines seem to outperform in a long run.

Exploring further, we compare the performance of the KT CS to the non-time-uniform confidence intervals (shown in dashed lines with transparency).
We first note that the asymptotic behavior of KT CS resembles
Sanov's theorem~\citep[Thm.~11.4.1]{Cover--Thomas2006}, which can be proved by the method of types: for each $t\ge 1$,
\begin{align}
\label{eq:sanov}
\P\Bigl\{D(\hat{\muv}_t~\|~\muv)
<\frac{1}{t}\log\frac{1}{\d} + \frac{K}{t}\log(t+1)\Bigr\}\ge 1-\d.
\end{align}
Note that it has a worse redundant term than KT CS. \citet{Mardia--Jiao--Tanczos--Nowak--Weissman2020} proposed several improved bounds of a similar flavor; for example, for each $t\ge 8\pi(K/e)^3$, 
\begin{align}
\label{eq:sanov_improved}
\P\Bigl\{D(\hat{\muv}_t~\|~\muv)
<\frac{K-1}{t}\log\frac{2(K-1)}{\d}\Bigr\}\ge 1-\d,
\end{align}
which removes the $O(\frac{1}{t}\log t)$ term.
Note that the required number of samples  for this improved bound to kick in is $t=\Omega(K^3)$ and grows fast in $K$. 
As shown in the simulation above, KT is almost always tighter than Sanov and on par after the improved bound kicks in, while KT CS provides a much stronger time-uniform guarantee.

We also compare with the combined Cloppe--Pearson intervals~\citep{Clopper--Pearson1934} with Bonferroni correction ($\mathsf{CP(2)+Bonf}$) and the Sison--Glaz method~\citep{Sison--Glaz1995}. 
We note that $\mathsf{CP(2)+Bonf}$ provides consistently tighter confidence sets with guaranteed coverage, but without the time-uniform guarantee\footnote{The \emph{coverage} (or \emph{coverage probability}) is the probability that the confidence set contains the target parameter. In the experiments, we compute the empirical coverage at each time instance, \ie non-time-uniform coverage, over the repeated experiments.}. While the Sison--Glaz method outputs extremely small confidence sets, the coverage is much below than the desired threshold, since the method was designed based on the asymptotic approximation.
This demonstrates that the performance of KT CS is even comparable to the existing non-time-uniform bounds.

\subsection{CS for Sampling Without Replacement}
\label{sec:wor}
In this section, we specifically consider the sampling \emph{without replacement} (WoR) setup, where we aim to construct a tight, vanishing confidence sequence to the mean of finite samples with sampling WoR. 
The main result in this section is that if one adapts the KT CS for this setting via a general reduction to be described below, the resulting CS is always tighter than the posterior-prior martingale proposed by \citet{Waudby-Smith--Ramdas2020a}. 

Suppose that there exists a finite set of $N$ multinomial observations $\{\zv_1,\ldots,\zv_N\}\subset\Vc_K$, which are fixed and nonrandom.
The underlying population is now fixed, and the goal here is to estimate the population mean ${\muv}\defeq \frac{1}{N}\sum_{i=1}^N \zv_i$ by sequentially drawing an observation from the population uniformly at random WoR. 
Formally, the probabilistic model for the sampling WoR procedure is 
\[
\Zv_t | \Zv^{t-1}\sim \Unif(\{\zv_1,\ldots, \zv_N\}\backslash \{\Zv_1,\ldots,\Zv_{t-1}\}).
\]
We note that this setup has important practical applications; for example, a tight WoR confidence sequence can lead to a fast decision in the real-world auditing for votes~\citep{Waudby-Smith--Stark--Ramdas2021}.
\newcommand{\wor}{\mathsf{WoR}}
As argued in \citep{Waudby-Smith--Ramdas2020b}, there exists a general reduction that maps any CS for bounded random processes to CS for sampling WoR.
For each $t\ge 1$, after having observed $\Zv^{t-1}$, the mean of the remaining samples is 
\[
\muv_{N,t}^{\wor}(\Zv^{t-1})
\defeq \E[\Zv_{t}|\Fc^{t-1}]
=\frac{N\muv-\sum_{i=1}^{t-1} \Zv_i}{N-(t-1)}.
\]
For any $\mv\in\Delta^{K-1}$, we define an analogous transform
\[
\mv_{N,t}^{\wor}(\Zv^{t-1})
\defeq\frac{N\mv-\sum_{i=1}^{t-1} \Zv_i}{N-(t-1)}.
\]
To testing if the underlying population mean is $\mv$, 
we can plug-in $\mv_{N,t+1}^{\wor}(\Zv^{t})$ in place of $\mv$ to the definition of any valid confidence sequence $(\Cc(\zv^t;\d))_{t=1}^\infty$ for bounded vectors; we defer the statement Proposition~\ref{prop:wor} to Appendix~\ref{app:sec:wor}.

Concretely, 
we can adapt the KT CS to this setup by only plugging-in $\mv_{N,t+1}^{\wor}(\Zv^t)$ in place of $\mv$ in the definition of the wealth process~\eqref{eq:kt_wealth}, which leads to
\begin{align*}
\wealth_{\wor}^\kt(\Zv^t;\mv)
&\defeq \wealth^\kt(\Zv^t;\mv_{N,t+1}^{\wor}(\Zv^t))\\
&= \frac{B(\Kv_t+\av)}{B(\av)} \Bigl(\frac{N\mv-\Kv_t}{N-t}\Bigr)^{-\Kv_t},
\end{align*}
where we define $\Kv_t\defeq \kv(\Zv^t)=\sum_{i=1}^t \Zv_i$ as a shorthand.
Here, we allow $\av\in\Real_{>0}^K$ to be a hyperparameter.

\citet{Waudby-Smith--Ramdas2020a} proposed a CS for categorical observations based on the process
\[
\mathsf{R}_{\wor}(\Zv^t;\mv)\defeq 
\frac{B(\Kv_t+\alphav)}{B(\alphav)} \frac{\binom{N}{N\mv}}{\binom{N-t}{N\mv-\Kv_t}}
\numberthis\label{eq:ppr}
\]
$\mv\in\Gc_{K,N}$, \ie $\mv\in \Integer_{\ge 0}^K$ and $\mv / N \in \Delta^{K-1}$,
where $\av\in\Real_{>0}^K$ is again a hyperparameter.
This sequence is derived as a posterior-prior ratio (PPR) under a hypergeometric observation model, and thus we refer to the resulting CS as PPR CS.
Interestingly, KT CS is always tighter than or equal to PPR CS with probability 1 as implied by the following theorem. 
Recall that a CS becomes tighter as the underlying sequence diverges faster.
\begin{theorem}
\label{thm:guarantee}
Pick any $0\le t\le N$ and $\tilde{\zv}^t\in\Vc_K^t$.
For any $\mv\in \Gc_{K,N}$, we have
\[
\wealth_{\wor}^\kt(\tilde{\zv}^t;\mv)\ge \mathsf{R}_{\wor}(\tilde{\zv}^t;\mv).
\]
\end{theorem}

\subsubsection{Simulation}
We demonstrate the tightness of the (WoR-) KT CS compared to the PPR CS in Fig.~\ref{fig:wor}. We consider a finite population of 1000 balls consisting of 600 red, 250 green, and 150 blue balls, and report the result averaged over 1000 random permutations. While both KT and PPR CSs naturally converge to the singleton containing the true population after observing all balls, KT can be relatively much tighter than PPR, as shown in the second panel. If one wishes to decide the rank of the colors by acquiring balls sequentially, the stopping time, which is defined as the first time when all colors can be distinguished from all the others with certainty, is the practical measure for the quality of CS. In the third panel, we visualize the histogram of the ratio of the stopping times between KT and PPR; on average, KT CS can decide the rank about 26\% faster than PPR CS.

\begin{figure*}
\centering
\includegraphics[width=\textwidth]{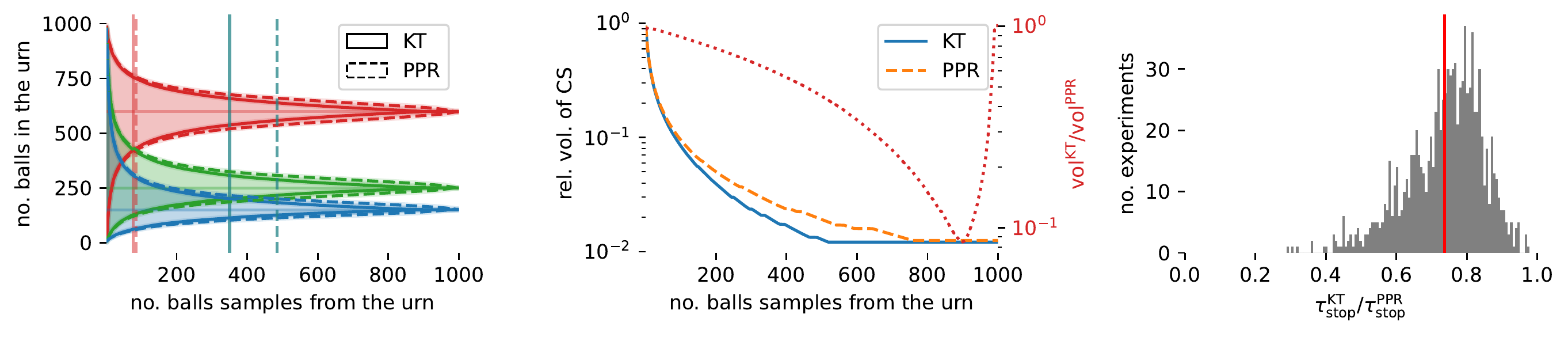}\vspace{-.5em}
\caption{Simulation of KT and PPR CSs for sampling without replacement (WoR). The underlying population consists of 1000 balls (600 red , 250 green, and 150 blue). 
The plotted results are averaged over random experiments with 1000 random permutations. 
The vertical lines in the first panel visualize the stopping time for each color when the confidence interval for the number of the balls of the particular color becomes disjoint from all the other confidence intervals.
Note, however, that the curves shown here are averaged over random trials and the stopping times indicated here are only for an illustrative purpose.
The second panel visualizes the relative volume of the CSs at each time step.
The red, dotted line in the second panel shows the ratio of the volume of KT CS to that of PPR CS.
The last panel presents a histogram of the ratios of the final stopping time of KT to PPR.}
\label{fig:wor}
\end{figure*}

\section{{Probability-Vector}-Valued Observations}
\label{sec:probability}
We now assume that $(\Yv_t)_{t=1}^{\infty}$ is a $\Delta^{K-1}$-valued stochastic process such that $\E[\Yv_t|\Fc_{t-1}]\equiv\muv$ for some $\muv\in\Delta^{K-1}$.
The development is in parallel Section~\ref{sec:categorical} for categorical observations. 
We only need to replace \hrk{}$(\mv^{-1})$ with \chrk{}$(\mv^{-1})$, where the odds vector $\Xv_t
\defeq {\mv^{-1}}\odot {\Yv_t}$ now can take any value in the convex hull of $\{\frac{\ev_1}{m_1},\ldots,\frac{\ev_K}{m_K}\}$.
With a slight abuse of notation, we also denote the cumulative wealth from the horse race \chrk{}$(\ov)$ by $\wealth(\Yv^t;\ov)$. 
It then readily follows that the same definition in \eqref{eq:kt_cs} is a $(1-\d)$-CS as shown in Theorem~\ref{thm:hr}:
\begin{theorem}
\label{thm:chr}
For any causal gambling strategy, $(\Cc(\Yv^t;\d))_{t=1}^{\infty}$ is a CS with level $1-\d$.
\end{theorem}

\subsection{Universal-Portfolio CS}
As previously for categorical observations and the HR-based CS construction, any causal gambling strategy leads to a valid CS.
In the previous section, we demonstrated some nice properties of the KT CS, which is induced by a universal gambling strategy based on the KT mixture.
Analogously, we will also explore the idea of universal gambling for probability-vector-valued observations. 
The idea of universal gambling in this continuous case is same: aiming to achieve the best wealth attained by any constant bettor, consider a mixture of the wealth of all constant bettors.
If we consider the KT mixture with $\Dir(\bv;\av)$ with $\av=\frac{\ones_K}{2}$, the universal gambling strategy is Cover's universal portfolio (UP) algorithm~\citep{Cover1991,Cover--Ordentlich1996}.
For the scalar case $(K=2)$, constructing a CS using UP was studied and its excellent empirical performance especially in a small-sample regime was demonstrated by \citep{Orabona--Jun2021,Ryu--Bhatt2022}. Here, we demonstrate a similar property in the multidimensional case.

\begin{figure*}[t]
    \centering
    \includegraphics[width=\textwidth]{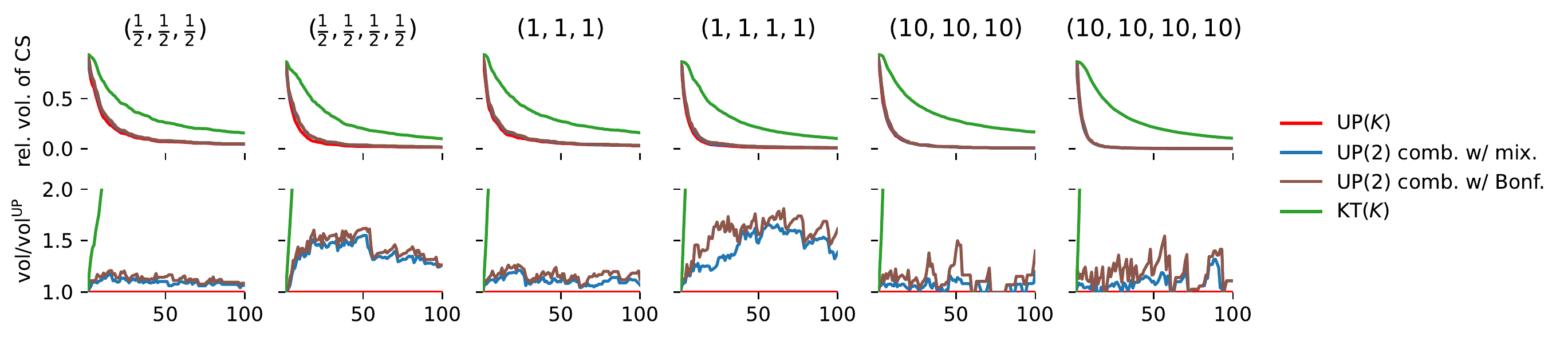}
    \caption{A similar experiment to Fig.~\ref{fig:ktcs} for \iid Dirichlet observations for $K\in\{3,4\}$.}
    \label{fig:upcs}
\end{figure*}

While the conceptual development is exactly the same, the difference arises in the definition of the wealth.
We start by considering the cumulative wealth attained by a constant bettor $\bv$ from \chrk{}$(\ov)$ for a $\Simplex_K$-valued sequence $\yv^t$:
\begin{align*}
{\wealth^{\bv}(\yv^t;\ov)}
&\defeq \prod_{i=1}^t \sum_{j=1}^K o_jb_{j}y_{tj}.
\numberthis\label{eq:wealth_const_betting}\\
&\stackrel{(\star)}{=} \sum_{\eta^t\in[K]^t} \prod_{i=1}^t (o_{\eta_i} b_{\eta_i}) y_{i\eta_i}\\
&= \sum_{\zv^t\in\Vc_K^t} \prod_{i=1}^t (\ov\odot\bv\odot \yv_i)^{\zv_i}
\\
&= \sum_{\kv\in\Gc_{K,t}} \yv^t[\kv] (\ov\odot\bv)^{\kv}.
\numberthis\label{eq:wealth_const_betting_chr}
\end{align*}
Here, noting that the wealth~\eqref{eq:wealth_const_betting} is a product of summations, we rewrite it as a summation of products by applying the distributive law in $(\star)$. We then further simplify the expression with the notations $\Vc_K$ and $\Gc_{K,t}$ defined in Table~\ref{tab:notation}. 
For $\kv\in\Gc_{K,t}$, we also define in \eqref{eq:wealth_const_betting_chr}
\[
\yv^t[\kv]\defeq 
\sum_{\zv^t\in\Vc_K^t\suchthat \kv(\zv^t)=\kv}
\prod_{i=1}^t \yv_i^{\zv_i}.
\]
Note that $\zv_i$'s are one-hot vectors and thus $\yv_i^{\zv_i}=\prod_{j=1}^K y_{i,z_{ij}}$ simply reads out the coordinate of $\yv_i$ where $\zv_i$ has 1.
As a sanity check, if $\yv\in\Vc_K$, the wealth simplifies to $\wealth^\bv(\yv^t;\ov)=(\ov\odot \bv)^{\kv(\yv^t)}$ as in \hrk{}$(\ov)$.

Now, if we consider the mixture of the wealth of constant bettors $\bv$ with respect to a Dirichlet distribution $\Dir(\bv;\av)$, the resulting wealth is the wealth achieved by Cover's UP:
\begin{align*}
{\wealth^{\mathsf{UP}}(\yv^t;\ob)}
&\defeq \int_{\Delta^{K-1}} {\wealth_t^\bv(\yv^t;\ov)} \Dir(\bv;\av)\diff\bv\\
&= \sum_{\kv\in\Gc_{K,t}} \yv^t[\kv] \ob^{\kv}
\frac{B(\kv+\av)}{B(\av)},
\end{align*}
since $\int_{\Delta^{K-1}}\bv^{\kv}\Dir(\bv;\av)\diff\bv=B(\kv+\av)/B(\av)$.
Note that the summation is over $\Gc_{K,t}$ which has cardinality $\binom{t+K-1}{t}=O(t^{K-1})$, and thus the summation can be computed in $O(t^{K-1})$ with memory complexity $O(t^{K-1})$ at each round $t$, by maintaining and updating the statistics $\yv^t[\kv]$ for $\kv\in\Gc_{K,t}$ by
\[
\yv^t[\kv]=
\sum_{j=1}^K y_{tj}\yv^{t-1}[\kv-\ev_j]\ind{k_j\ge 1}.
\numberthis
\label{eq:up_recursive_update}
\]
after observing $\yv_t$ at round $t\ge 1$.
We note that this dynamic programming argument was essentially discussed in \citet[Section~VI]{Cover--Ordentlich1996}.

Let $(\Cc^{\up}(\Yv^t;\d))_{t=1}^\infty$ denote the resulting $(1-\d)$-CS from the UP wealth $\wealth^{\up}(\yv^t;\mv^{-1})$, \ie
\[
\Cc^\up(\yv^t;\d)\defeq \Bigl\{
\mv\in\Simplex_K
\suchthat {\wealth(\yv^t;\mv^{-1})} <\frac{1}{\d}
\Bigr\}.
\numberthis\label{eq:up_cs}
\]
We formally state the properties of UP CS as in Theorem~\ref{thm:kt}:
\begin{theorem}[UP CS]
\label{thm:up}
Pick any $\d\in(0,1)$.
\begin{enumerate}[label={(\alph*)}]
\item (Convexity) ${\Cc}_t^{\up}(\yv^t;\d)\subset\Delta^{K-1}$
is always convex.
\item (Always-non-vacuous guarantee) For any $j\in[K]$, $\Cc^{\up}(\yv^t;\d)$ is strictly bounded away from the hyperplane $\{\mv\suchthat m_j=0\}$ as soon as $\max_{i\in[t]}y_{ij}>0$.
\item $\hat{\muv}_t\in {\Cc}_t^{\up}(\Yv^t)$ with probability 1.
\end{enumerate}
\end{theorem}
\citet{Orabona--Jun2021} established a similar property to (b) for UP CS in the scalar-valued case; (c) is a generalization of \citet[Theorem~4.3(c)]{Ryu--Bhatt2022}.

\subsubsection{Simulation} 
We conducted a simulation similar to Fig.~\ref{fig:ktcs}, but for \iid Dirichlet observations of dimension $K\in\{3,4\}$ with different concentration parameters. 
As baselines, we tested the aggregated $\mathsf{UP}(2)$ CSs with mixture and Bonferroni correction, as well as $\mathsf{KT}(K)$. Compared to $\mathsf{KT}(K)$, $\mathsf{UP}(K)$ is much tighter in general, as the distribution gets more concentrated (small variance). Similar to the observation in Fig.~\ref{fig:ktcs}, the gap between $\mathsf{UP(2)~w/~mix.}$ and $\mathsf{UP}(K)$ increases as the dimension $K$ increases.

\subsection{Reduction for Bounded Vectors}
\label{sec:general}
Finally, the technique developed so far can be applied to construct a CS for $[0,1]^{K-1}$-valued observations for $K\ge 2$ by a simple reduction described below.
Suppose that we observe $\Yv_t=(Y_{t1},\ldots,Y_{t,K-1})\in[0,1]^{K-1}$ for each $t\ge1$, which satisfies $\E[\Yv_t|\Fc_{t-1}]\equiv\muv$ for some $\muv\in[0,1]^{K-1}$.
To construct a confidence sequence for $\muv$, we can consider its probability-vector representation $\tilde{\Yv}_t=(\Yt_{t1},\ldots,\Yt_{tK})\in\Delta^{K}$, which is defined as
as
\begin{align}\label{eq:reduction}
\Yt_{tK}\defeq 1-\frac{1}{K-1}(Y_{t1}+\ldots+Y_{t,K-1})
\end{align}
and $\Yt_{tj}\defeq Y_{tj}/(K-1)$ for $j\in[K-1]$.
Note that $\tilde\Yv_t$ lies on the $K$-dimensional probability simplex $\Delta^{K-1}$ by construction and now its conditional mean is \[
\E[\tilde{\Yv}_t|\Fc_{t-1}]=\begin{bmatrix}
\frac{\muv}{K-1} \\
1-\frac{\ones_{K-1}^\intercal\muv}{K-1}
\end{bmatrix}\in\Delta^{K-1}.
\]
Now we can apply the UP CS on the sequence $(\tilde{\Yv}_t)_{t=1}^\infty$ to construct a confidence sequence for $\muv$.

\section{Concluding Remarks}
\label{sec:conclusion}
In this paper, we presented a general gambling framework for constructing confidence sequences and established some theoretical properties of CSs from universal gambling strategies.
Our simulations show that compared to the baseline CSs constructed with Bonferroni correction and averaged wealth, the multivariate gambling CSs are more conservative (\ie of higher coverage for each time instance), provide tighter confidence sets in the small-sample regime, and the performance margin increases as the dimensionality increases.
Such a tight confidence set construction may lead to faster decisions in sequential decision-making problems, \ie , as we demonstrated in Section~\ref{sec:wor}.
We remark that one can consider applying the KT CS for multi-arm bandits as demonstrated in \citep[Section~4.4]{Malloy--Tripathy--Nowak2020}.

While UP CS provides extremely tight confidence sets with guaranteed coverage even with small samples, the prohibitive computational complexity $O(T^K)$ for $K$-dimensional probability vectors of length $T$ makes UP CS less applicable in practice for high-dimensional data. 
Efficient approximation of UP has been an open problem in information theory and theoretical computer science for decades~\citep{VanErven--VanderHoeven--Kotlowski--Koolen2020}.
For example, \citep{Kalai--Vempala2002} proposed efficient algorithms for approximate computation of UP based on Monte Carlo Markov chain methods, and several developments have been made recently, including \citep{Luo--Wei--Zheng2018, VanErven--VanderHoeven--Kotlowski--Koolen2020, Mhammedi--Rakhlin2022, Zimmert--Agarwal--Kale2022,Jezequel--Ostrovskii--Gaillard2022,Tsai--Lin--Li2024}.
The exploration of efficient approximate UP methods and the computational-statistical trade-off is left for future research. In the context of CSs, finding an efficient yet tight approximation of UP in the small-sample regime would be of great practical importance.

An experienced reader may wonder if the proposed gambling technique is applicable to sequential kernel regression (or kernel bandit)\citep{Vakili--Scarlett--Javidi2021}, or even to simpler linear bandits\citep{Abbasi-Yadkori--Pal--Szepesvari2011}. While this extension is intriguing, nontrivial technicalities arise. First, the proposed methods in this paper only apply to bounded vectors, requiring a suitable problem setting in such an extension. Second, the UP-based confidence sets are implicitly defined, needing a tight analytical outer bound or an efficient numerical procedure for implementation. We leave this direction for future work.

\section*{Acknowledgements}
The authors appreciate insightful feedback and suggestions from anonymous
reviewers to improve the manuscript.
This work was supported, in part, by the MIT-IBM Watson AI Lab under Agreement No. W1771646, and by MIT Lincoln Laboratory.

\section*{Impact Statement}
This paper presents work whose goal is to advance the field of Sequential Decision Making. 
There are many potential societal consequences of our work, none of which we feel must be specifically highlighted here.

\bibliographystyle{icml2024}
\bibliography{ref}

\newcommand{\noopsort}[1]{}
\begin{thebibliography}{38}
\providecommand{\natexlab}[1]{#1}
\providecommand{\url}[1]{\texttt{#1}}
\expandafter\ifx\csname urlstyle\endcsname\relax
  \providecommand{\doi}[1]{doi: #1}\else
  \providecommand{\doi}{doi: \begingroup \urlstyle{rm}\Url}\fi

\bibitem[Abbasi-Yadkori et~al.(2011)Abbasi-Yadkori, P{\'a}l, and Szepesv{\'a}ri]{Abbasi-Yadkori--Pal--Szepesvari2011}
Abbasi-Yadkori, Y., P{\'a}l, D., and Szepesv{\'a}ri, C.
\newblock Improved algorithms for linear stochastic bandits.
\newblock In \emph{Adv. Neural Inf. Proc. Syst.}, volume~24, 2011.

\bibitem[Cho et~al.(2024)Cho, Gan, and Kallus]{Cho--Gan--Kallus2024}
Cho, B., Gan, K., and Kallus, N.
\newblock Peeking with {PEAK}: Sequential, nonparametric composite hypothesis tests for means of multiple data streams.
\newblock \emph{arXiv preprint arXiv:2402.06122}, 2024.

\bibitem[Chugg et~al.(2023)Chugg, Wang, and Ramdas]{Chugg--Wang--Ramdas2023}
Chugg, B., Wang, H., and Ramdas, A.
\newblock Time-uniform confidence spheres for means of random vectors.
\newblock \emph{arXiv preprint arXiv:2311.08168}, 2023.

\bibitem[Clopper \& Pearson(1934)Clopper and Pearson]{Clopper--Pearson1934}
Clopper, C.~J. and Pearson, E.~S.
\newblock The use of confidence or fiducial limits illustrated in the case of the binomial.
\newblock \emph{Biometrika}, 26\penalty0 (4):\penalty0 404--413, 1934.

\bibitem[Cover(1991)]{Cover1991}
Cover, T.~M.
\newblock Universal portfolios.
\newblock \emph{Math. Financ.}, 1\penalty0 (1):\penalty0 1--29, 1991.

\bibitem[Cover \& Ordentlich(1996)Cover and Ordentlich]{Cover--Ordentlich1996}
Cover, T.~M. and Ordentlich, E.
\newblock Universal portfolios with side information.
\newblock \emph{{IEEE} Trans. Inf. Theory}, 42\penalty0 (2):\penalty0 348--363, 1996.

\bibitem[Cover \& Thomas(2006)Cover and Thomas]{Cover--Thomas2006}
Cover, T.~M. and Thomas, J.~A.
\newblock \emph{Elements of information theory}.
\newblock John Wiley \& Sons, 2006.

\bibitem[Darling \& Robbins(1967)Darling and Robbins]{Darling--Robbins1967}
Darling, D.~A. and Robbins, H.
\newblock Confidence sequences for mean, variance, and median.
\newblock \emph{Proc. Natl. Acad. Sci. U. S. A.}, 58\penalty0 (1):\penalty0 66, 1967.

\bibitem[Evidently(2023)]{Evidently2023}
Evidently.
\newblock How evidently calculates results, 2023.
\newblock URL \url{https://docs.aws.amazon.com/AmazonCloudWatch/latest/monitoring/CloudWatch-Evidently-calculate-results.html}.

\bibitem[Hoeffding(1963)]{Hoeffding1963}
Hoeffding, W.
\newblock Probability inequalities for sums of bounded random variables.
\newblock \emph{J. Am. Stat. Assoc.}, 58\penalty0 (301):\penalty0 13--30, 1963.

\bibitem[Howard et~al.(2021)Howard, Ramdas, McAuliffe, and Sekhon]{Howard--Ramdas--McAuliffe--Sekhon2021}
Howard, S.~R., Ramdas, A., McAuliffe, J., and Sekhon, J.
\newblock Time-uniform, nonparametric, nonasymptotic confidence sequences.
\newblock \emph{Ann. Statist.}, 49\penalty0 (2):\penalty0 1055--1080, 2021.

\bibitem[J{\'e}z{\'e}quel et~al.(2022)J{\'e}z{\'e}quel, Ostrovskii, and Gaillard]{Jezequel--Ostrovskii--Gaillard2022}
J{\'e}z{\'e}quel, R., Ostrovskii, D.~M., and Gaillard, P.
\newblock Efficient and near-optimal online portfolio selection.
\newblock \emph{arXiv e-prints}, pp.\  arXiv--2209, 2022.

\bibitem[Jun \& Orabona(2019)Jun and Orabona]{Jun--Orabona2019}
Jun, K.-S. and Orabona, F.
\newblock Parameter-free online convex optimization with sub-exponential noise.
\newblock In \emph{Conf. Learn. Theory}, pp.\  1802--1823. PMLR, 2019.

\bibitem[Kalai \& Vempala(2002)Kalai and Vempala]{Kalai--Vempala2002}
Kalai, A.~T. and Vempala, S.
\newblock Efficient algorithms for universal portfolios.
\newblock \emph{J. Mach. Learn. Res.}, pp.\  423--440, 2002.

\bibitem[Krichevsky \& Trofimov(1981)Krichevsky and Trofimov]{Krichevsky--Trofimov1981}
Krichevsky, R. and Trofimov, V.
\newblock The performance of universal encoding.
\newblock \emph{{IEEE} Trans. Inf. Theory}, 27\penalty0 (2):\penalty0 199--207, 1981.

\bibitem[Lai(1976)]{Lai1976}
Lai, T.~L.
\newblock On confidence sequences.
\newblock \emph{Ann. Statist.}, 4\penalty0 (2):\penalty0 265--280, 1976.

\bibitem[Li et~al.(2023)Li, Li, and Dai]{Li--Li--Dai2023}
Li, J., Li, Y., and Dai, X.
\newblock {{Li, Li, and Dai's Contribution to the Discussion of ``Estimating Means of Bounded Random Variables by Betting'' by Waudby-Smith and Ramdas}}.
\newblock \emph{J. R. Stat. Soc. Series B Stat. Methodol.}, pp.\  qkad111, October 2023.
\newblock ISSN 1369-7412.
\newblock \doi{10.1093/jrsssb/qkad111}.

\bibitem[Luo et~al.(2018)Luo, Wei, and Zheng]{Luo--Wei--Zheng2018}
Luo, H., Wei, C.-Y., and Zheng, K.
\newblock Efficient online portfolio with logarithmic regret.
\newblock In \emph{Adv. Neural Inf. Proc. Syst.}, volume~31, 2018.

\bibitem[Malloy et~al.(2020)Malloy, Tripathy, and Nowak]{Malloy--Tripathy--Nowak2020}
Malloy, M.~L., Tripathy, A., and Nowak, R.~D.
\newblock {Optimal Confidence Regions for the Multinomial Parameter}.
\newblock \emph{arXiv}, February 2020.

\bibitem[Manole \& Ramdas(2023)Manole and Ramdas]{Manole--Ramdas2023}
Manole, T. and Ramdas, A.
\newblock Martingale methods for sequential estimation of convex functionals and divergences.
\newblock \emph{{IEEE} Trans. Inf. Theory}, 69\penalty0 (7):\penalty0 4641--4658, 2023.
\newblock \doi{10.1109/TIT.2023.3250099}.

\bibitem[Mardia et~al.(2020)Mardia, Jiao, T{\'a}nczos, Nowak, and Weissman]{Mardia--Jiao--Tanczos--Nowak--Weissman2020}
Mardia, J., Jiao, J., T{\'a}nczos, E., Nowak, R.~D., and Weissman, T.
\newblock Concentration inequalities for the empirical distribution of discrete distributions: beyond the method of types.
\newblock \emph{Inf. Inference}, 9\penalty0 (4):\penalty0 813--850, 2020.

\bibitem[Mhammedi \& Rakhlin(2022)Mhammedi and Rakhlin]{Mhammedi--Rakhlin2022}
Mhammedi, Z. and Rakhlin, A.
\newblock Damped online {N}ewton step for portfolio selection.
\newblock \emph{arXiv preprint arXiv:2202.07574}, 2022.

\bibitem[Orabona \& Jun(2024)Orabona and Jun]{Orabona--Jun2021}
Orabona, F. and Jun, K.-S.
\newblock Tight concentrations and confidence sequences from the regret of universal portfolio.
\newblock \emph{{IEEE} Trans. Inf. Theory}, 70\penalty0 (1):\penalty0 436--455, 2024.
\newblock \doi{10.1109/TIT.2023.3330187}.
\newblock arXiv:2110.14099.

\bibitem[Ramdas et~al.(2020)Ramdas, Ruf, Larsson, and Koolen]{Ramdas--Ruf--Larsson--Koolen2020}
Ramdas, A., Ruf, J., Larsson, M., and Koolen, W.
\newblock Admissible anytime-valid sequential inference must rely on nonnegative martingales.
\newblock \emph{arXiv preprint arXiv:2009.03167}, September 2020.

\bibitem[Ryu \& Bhatt(2022)Ryu and Bhatt]{Ryu--Bhatt2022}
Ryu, J.~J. and Bhatt, A.
\newblock On confidence sequences for bounded random processes via universal gambling strategies.
\newblock \emph{arXiv preprint arXiv:2207.12382}, 2022.

\bibitem[Shekhar \& Ramdas(2023)Shekhar and Ramdas]{Shekhar--Ramdas2023}
Shekhar, S. and Ramdas, A.
\newblock On the near-optimality of betting confidence sets for bounded means.
\newblock \emph{arXiv preprint arXiv:2310.01547}, 2023.

\bibitem[Sison \& Glaz(1995)Sison and Glaz]{Sison--Glaz1995}
Sison, C.~P. and Glaz, J.
\newblock Simultaneous confidence intervals and sample size determination for multinomial proportions.
\newblock \emph{J. Am. Statist. Assoc.}, 90\penalty0 (429):\penalty0 366--369, 1995.

\bibitem[Tsai et~al.(2024)Tsai, Lin, and Li]{Tsai--Lin--Li2024}
Tsai, C.-E., Lin, Y.-T., and Li, Y.-H.
\newblock Data-dependent bounds for online portfolio selection without {L}ipschitzness and smoothness.
\newblock In \emph{Adv. Neural Inf. Proc. Syst.}, volume~36, 2024.

\bibitem[Vakili et~al.(2021)Vakili, Scarlett, and Javidi]{Vakili--Scarlett--Javidi2021}
Vakili, S., Scarlett, J., and Javidi, T.
\newblock Open problem: Tight online confidence intervals for {RKHS} elements.
\newblock In \emph{Conf. Learn. Theory}, pp.\  4647--4652. PMLR, 2021.

\bibitem[van Erven et~al.(2020)van Erven, Van~der Hoeven, Kot{\l}owski, and Koolen]{VanErven--VanderHoeven--Kotlowski--Koolen2020}
van Erven, T., Van~der Hoeven, D., Kot{\l}owski, W., and Koolen, W.~M.
\newblock Open problem: {F}ast and optimal online portfolio selection.
\newblock In \emph{Conf. Learn. Theory}, pp.\  3864--3869. PMLR, 2020.

\bibitem[Ville(1939)]{Ville1939}
Ville, J.
\newblock Etude critique de la notion de collectif.
\newblock \emph{Bull. Amer. Math. Soc}, 45\penalty0 (11):\penalty0 824, 1939.

\bibitem[Waudby-Smith \& Ramdas(2020)Waudby-Smith and Ramdas]{Waudby-Smith--Ramdas2020a}
Waudby-Smith, I. and Ramdas, A.
\newblock Confidence sequences for sampling without replacement.
\newblock In Larochelle, H., Ranzato, M., Hadsell, R., Balcan, M., and Lin, H. (eds.), \emph{Adv. Neural Inf. Proc. Syst.}, volume~33, pp.\  20204--20214. Curran Associates, Inc., 2020.

\bibitem[Waudby-Smith \& Ramdas(2024)Waudby-Smith and Ramdas]{Waudby-Smith--Ramdas2020b}
Waudby-Smith, I. and Ramdas, A.
\newblock Estimating means of bounded random variables by betting.
\newblock \emph{J. R. Stat. Soc. B}, 86\penalty0 (1):\penalty0 1--27, 2024.

\bibitem[Waudby-Smith et~al.(2021)Waudby-Smith, Stark, and Ramdas]{Waudby-Smith--Stark--Ramdas2021}
Waudby-Smith, I., Stark, P., and Ramdas, A.
\newblock {RiLACS: Risk Limiting Audits via Confidence Sequences}.
\newblock \emph{E-Vote-ID 2021}, pp.\  130, 2021.

\bibitem[Whitehouse et~al.(2023)Whitehouse, Wu, and Ramdas]{Whitehouse--Wu--Ramdas2023}
Whitehouse, J., Wu, Z.~S., and Ramdas, A.
\newblock Time-uniform self-normalized concentration for vector-valued processes.
\newblock \emph{arXiv preprint arXiv:2310.09100}, 2023.

\bibitem[Xie \& Barron(2000)Xie and Barron]{Xie--Barron2000}
Xie, Q. and Barron, A.~R.
\newblock Asymptotic minimax regret for data compression, gambling, and prediction.
\newblock \emph{{IEEE} Trans. Inf. Theory}, 46\penalty0 (2):\penalty0 431--445, 2000.

\bibitem[Yang \& Berger(1996)Yang and Berger]{Yang--Berger1996}
Yang, R. and Berger, J.~O.
\newblock \emph{A catalog of noninformative priors}, volume 1018.
\newblock Institute of Statistics and Decision Sciences, Duke University Durham, NC, USA, 1996.

\bibitem[Zimmert et~al.(2022)Zimmert, Agarwal, and Kale]{Zimmert--Agarwal--Kale2022}
Zimmert, J., Agarwal, N., and Kale, S.
\newblock Pushing the efficiency-regret {P}areto frontier for online learning of portfolios and quantum states.
\newblock \emph{arXiv preprint arXiv:2202.02765}, 2022.

\end{thebibliography}

\newpage
\appendix
\onecolumn

\section{Deferred Proofs}
\label{app:proofs}
\subsection{Proof of Theorem~\ref{thm:kt}}
\begin{proof}
The convexity of the KT CS in (b) follows from the observation that $\mv\mapsto \av^\intercal \log\mv^{-1}$ is convex for any $\av\in\Real_{\ge0}^K$, and the sublevel set of a convex function is always convex.
To prove (c), observe that
\begin{align*}
\Cc^{\kt}(\zv^t;\d)
&\stackrel{(i)}{\subset} \Bigl\{\mv\in\Delta^{K-1}\suchthat \hat{\muv}_t^\intercal\mv^{-1} < (\d q^{\kt}(\kv(\zv^t)))^{-\frac{1}{t}}\Bigr\}\\
&\stackrel{(ii)}{\subset} \Bigl\{\mv\in\Delta^{K-1}\suchthat m_j > \frac{k_j(\zv^t)}{t}(\d q^{\kt}(\kv(\zv^t)))^{\frac{1}{t}}\Bigr\},
\end{align*}
where $(i)$ follows from Jensen's inequality and $(ii)$ holds by deploying 
the lower bound $\hat{\muv}_t^\intercal\mv^{-1} \ge \frac{k_j(\zv^t)}{t}\frac{1}{m_j}$ for any $j\in[K]$.
The last property (d) follows from noting that by Stirling's formula, the second term in the definition of the CS behaves as
\[
\frac{1}{t}\log\frac{e^{-t H(\hat{\muv}_t})}{q^{\kt}(t\hat{\muv}_t)}\sim \frac{K-1}{2t}\log t + o(1)
\]
for $t$ sufficiently large.
\end{proof}

\subsection{Proof of Theorem~\ref{thm:guarantee}}
The proof relies on a simple combinatorial inequality.
It suffices to show that 
\[
\binom{K}{\kb} 
\ge 
\binom{K+M}{\kb+\mv} 
\Bigl(\frac{\kb}{K}\Bigr)^{\mv}
\]
for $K,M\in\Natural$, $\kb,\mv\in\Natural_+^{m}$ such that $\sum_{i=1}^m k_i=K$ and $\sum_{i=1}^m m_i = M$, since the desired inequality follows by plugging in
$K\gets N-t$, $\kb\gets \nv-\Kv_t$, $M\gets t$, $\mv\gets \Kv_t$.
The desired inequality follows from Lemma~\ref{lem:general_inequality}, which is stated and proved below.
\qed

\begin{lemma}
\label{lem:general_inequality}
Let $k_1,\ldots,k_d\ge 1$ and $m_1,\ldots,m_d\ge 0$. Then,
\[
\frac{(\sum_{i=1}^d k_i)! (\sum_{i=1}^d k_i)^{\sum_{i=1}^d m_i}}{(\sum_{i=1}^d(k_i+m_i))!} 
\ge \prod_{i=1}^d \frac{k_i! k_i^{m_i}}{(k_i+m_i)!}.
\]
\end{lemma}

\begin{proof}
It suffices to show for $d=2$, \ie
\begin{align}
\label{eq:simple_ineq}
\binom{k_1+k_2}{k_1} 
\ge 
\binom{k_1+k_2+m_1+m_2}{k_1+m_1} 
\frac{k_1^{m_1} k_2^{m_2}}{(k_1+k_2)^{m_1+m_2}}.
\end{align}
Rearranging the terms, we can write
\[
\binom{k_1+k_2}{k_1} \Bigl(\frac{k_1}{k_1+k_2}\Bigr)^{k_1} \Bigl(\frac{k_2}{k_1+k_2}\Bigr)^{k_2}
\ge 
\binom{k_1+k_2+m_1+m_2}{k_1+m_1}  \Bigl(\frac{k_1}{k_1+k_2}\Bigr)^{k_1+m_1} \Bigl(\frac{k_2}{k_1+k_2}\Bigr)^{k_2+m_2}.
\]
If we consider two independent binomial random variables $Z\sim \Binom(k_1+k_2,\frac{k_1}{k_1+k_2})$ and $W\sim \Binom(m_1+m_2,\frac{k_1}{k_1+k_2})$, then the inequality is equivalent to 
\[
\P(Z=k_1) \ge \P(Z+W=k_1+m_1),
\]
since $Z+W\sim \Binom(k_1+k_2+m_1+m_2, \frac{k_1}{k_1+k_2})$.
Finally, the inequality readily follows from the following chain of inequalities:
\begin{align*}
\P(Z+W=k_1+m_1)
&= \sum_{0\le z\le \min\{k_1+m_1, k_1+k_2\}} \P(Z=z) \P(W=k_1+m_1-z)\\
&\le \P(Z=k_1) \sum_{0\le z\le \min\{k_1+m_1, k_1+k_2\}} \P(W=k_1+m_1-z)\\
&\le \P(Z=k_1).
\end{align*}
Here, we used the fact that $\max_{0\le z\le k_1+k_2} \P(Z=z) = \P(Z=k_1)$ in the first inequality. This completes the proof.
\end{proof}

\subsection{Proof of Theorem~\ref{thm:up}}
\begin{proof}
The convexity of UP CS follows from the fact that $\mv\mapsto \mv^{-\kv}$ is log-convex for each $\kv\in\Gc_{K,t}$, a sum of any log-convex functions is also log-convex. 
The proof for (b) is similar to that of Theorem~\ref{thm:kt}(b) and thus omitted.
We can prove (c) via the Lemma~\ref{lem:small_wealth} stated below, which states that no constant bettor can earn any money from \chrk{}$(\hat{\muv}_t^{-1})$ for any outcome sequence $\yv^t$, \emph{deterministically}.
\end{proof}

\begin{lemma}
\label{lem:small_wealth}
For $\yv^t\in(\Delta^{K-1})^t$, let $\hat{\muv}_t=\frac{1}{t}\sum_{i=1}^t \yv_i$.
For any $\bv\in\Delta^{K-1}$, we have
\[
{\wealth_t^\bv(\yv^t;{\hat{\muv}_t}^{-1})}
\le 1.
\]
\end{lemma}
\begin{proof}
Note that, from \eqref{eq:wealth_const_betting}, we can write
\[
\wealth_t^\bv(\yv^t;\hat{\muv}_t^{-1})
=\sum_{\kv\in\Gc_t}  \frac{\bv^{\kv}}{\hat{\muv}_t^{\kv}} \yv^t[\kv].
\]
Since it is easy to verify that $\wealth_t^{\bv=\hat{\muv}_t}(\yv^t;\hat{\muv}_t^{-1})=1$, we only need to show that the global maximum of the function $\bv\mapsto \wealth_t^{\bv}(\yv^t;\hat{\muv}_t^{-1})$ is attained when $\bv=\hat{\muv}_t$.
Note that the function $\bv\mapsto \wealth_t^\bv(\yv^t;\hat{\muv}_t^{-1})$ is log-concave, since $\bv\mapsto \bv^{\kv}$ is log-concave and a sum of any log-concave functions is also log-concave.
Moreover, it is easy to check that the gradient of the function with respect to $\bv$ at $\bv=\hat{\muv}_t$ is perpendicular to the probability simplex $\Delta_m$.
That is, we have
\[
\frac{\partial}{\partial b_j}\wealth_t^\bv(\yv^t;\hat{\muv}_t^{-1})|_{\bv=\hat{\muv}_t}
= \sum_{\kv\in\Gc_t}\yv^t[\kv]\Bigl(\frac{k_j}{\hat{\mu}_{tj}}-\sum_{j'\in[m]\backslash\{j\}} \frac{k_{j'}}{\hat{\mu}_{tj'}}\Bigr)
=-(m-2)t,
\]
where the last equality holds since $\sum_{\kv\in\Gc_t} k_j\yv^t[\kv]=\sum_{i=1}^t y_{ij} =t\hat{\mu}_{tj}$.
Hence, the gradient \[\nabla_{\bv}\wealth_t^\bv(\yv^t;\hat{\muv}_t^{-1})|_{\bv=\hat{\muv}_t}=-(m-2)t\ones_K\] is perpendicular to the probability simplex $\Delta^{K-1}$, and thus together with the log-concavity, we can conclude that $\bv=\hat{\muv}_t$ is the global maximizer.
\end{proof}

\section{General Reduction for CS with WoR Sampling}
\label{app:sec:wor}
\begin{proposition}[\citet{Waudby-Smith--Ramdas2020a}]
\label{prop:wor}
Suppose that a sequence $\{W_t(\tilde{\Zv}^t;\mv)\}_{t=0}^\infty$ with $\wealth_0(\varphi;\mv)=1$ is a martingale for $\mv=\tilde{\muv}$, if a stochastic process $(\tilde{\Zv})_{t=1}^\infty$ satisfies $\E[\tilde{\Zv}_t|\Fc_{t-1}]=\tilde{\muv}$.
If we define
\begin{align*}
\Cc_{\wor}(\Zv^t;\d)
&\defeq \Bigl\{\mv\suchthat
W_t(\Zv^t;\mv_{N,t+1}^{\wor}(\Zv^t))<\frac{1}{\d}
\Bigr\}\\
&= \Bigl\{\mv\suchthat
W_t\Bigl(\Zv^t;\frac{N\mv-\sum_{i=1}^{t} \Zv_i}{N-t}\Bigr)<\frac{1}{\d}
\Bigr\}
\end{align*}
for $t\in[N-1]$, then 
\[
\P(\forall t\in[N-1], \muv\in \Cc_{\wor}(\Zv^t;\d)) \ge 1-\d,
\]
given that $(\Zv_t)_{t=1}^N$ is sampled WoR from $\{\zv_i\}_{i=1}^N$.
\end{proposition}

\section{Posterior-Prior Ratio Interpretation of PPR CS}
\label{app:wor_ppr}

As alluded to earlier, the PPR sequence~\eqref{eq:ppr} can be derived as a posterior-prior ratio as follows.
Note that we can write
\[
\mathsf{R}(\Yv^t;\mv)= \frac{\pi_0(N\mv)}{\pi_t(N\mv)},
\]
if we define
\[
\pi_t(\nv)\defeq 
\frac{B(\nv+\alphav)}{B(\Kv_t+\alphav)}
\binom{N-t}{\nv-\Kv_t}
\]
for $t\ge 0$ and $\nv\in\Gc_{N,K}$.
This quantity $\pi_t(\nv)$ is the Dirichlet multinomial distribution, which is the posterior distribution induced by the hypergeometric observation model, while $\pi_0(\nv)$ corresponds to the prior distribution.
More precisely, if we assume $\Yv_t|(\Yv_1,\ldots,\Yv_{t-1})\sim \mathsf{MultHyperGeo}(N-(t-1), \nv-\sv_{t-1}, 1)$ and $\nv\sim \mathsf{DirMult}(N,\av)$. 
Then, by the conjugacy, one can show that $\nv-\Kv_t|(\Yv_1,\ldots,\Yv_t)\sim \mathsf{DirMult}(N-t,\Kv_t+\av)$, which leads to the definition of $\pi_t(\nv)$; see \citep[Theorem C.1]{Waudby-Smith--Ramdas2020a}.
It can be shown that any prior-to-posterior ratio sequence is a martingale if $\mv=\muv$, and thus can be used to construct a CS.

\end{document}